\documentclass[a4paper, 10pt]{article}
\usepackage{geometry}

\usepackage{natbib}

\usepackage{filecontents}
\usepackage{float}

\usepackage[english]{babel}
\usepackage[utf8]{inputenc}
\usepackage[T1]{fontenc}

\usepackage{amsmath}
\usepackage{amssymb}

\usepackage{xcolor}
\usepackage{bbm}

\usepackage{enumerate}
\usepackage{scalerel} 
\usepackage{subfigure}

\usepackage{graphicx}

\usepackage{amsthm}
\newtheorem{theorem}{Theorem}
\numberwithin{theorem}{section}
\newtheorem{lemma}[theorem]{Lemma}
\newtheorem{prop}[theorem]{Proposition}

\newtheorem{corollary}[theorem]{Corollary}

\theoremstyle{definition}
\newtheorem{definition}[theorem]{Definition}
\newtheorem{example}[theorem]{Example}
\newtheorem{remark}[theorem]{Remark}

\usepackage{pgfplots}

\pgfplotsset{compat = newest}

\newcommand{\N}{\mathbb{N}}

\newcommand{\E}{\mathbb{E}}

\newcommand{\A}{\mathcal{A}}

\newcommand{\F}{\mathcal{F}}
\newcommand{\G}{\mathcal{G}}

\newcommand{\ind}[1]{\mathbf{1}_{#1}}

\newcommand{\Var}{\mathrm{Var}}

\renewcommand{\epsilon}{\varepsilon}

\renewcommand{\P}{\mathbb{P}}

\let\oldmarginpar\marginpar
\renewcommand\marginpar[1]{\-\oldmarginpar[\raggedleft\footnotesize #1]{\raggedright\footnotesize\color{red} #1}}
\usepackage{hyperref}

\usepackage{tikz}
\usepackage{subfigure}
\usepackage{pgfplots}
\pgfplotsset{compat=newest}

\title{Delta-Epsilon-Common Knowledge and Quantitative Agreement Theorems}

\author{Christina Pawlowitsch\footnote{University Paris--Panthéon--Assas, Laboratoire de Mathématique Économique, and Institut Léon Walras: christina.pawlowitsch@assas-universite.fr}, Stefan Schrott\footnote{University of Münster, Institute for Mathematical Stochastics: stefan.schrott@uni-muenster.de}, Daniel Toneian\footnote{University of Vienna, Faculty of Mathematics: daniel.toneian@univie.ac.at}}

\begin{document}
	\maketitle
	\begin{abstract}
		
	Aumann \cite{Au76} defined common knowledge mathematically and established his now famous Agreement Theorem. We present a novel approach to quantifying how close individuals are to commonly knowing events, $(\delta,\epsilon)$-common knowledge, which is defined for any (and not just countable) probability spaces, and provide quantitative versions of the key results in this field. Specifically, we do this for Aumann's \cite{Au76}  Agreement Theorem and Nielsen's \cite{Ni84} extension thereof to random variables, as well as for the setting in which posteriors are communicated back and forth between individuals. Our results apply in particular to noisy communication settings.
		\end{abstract}
	
	\section{Introduction}
	
Fifty years after its publication, Aumann's \cite{Au76} Agreement Theorem and the notion of \emph{common knowledge}  introduced therein have not ceased to intrigue and inspire (see, for instance, recent contributions by Billot and Vergopoulos~\cite{BiVe26}, Di Tillio et al.~\cite{DiLeSa22}, Gizatulina and Hellman~\cite{GiHe19}, Hellman and Pinter~\cite{HePi22}, Geanakoplos and Polemarchakis~\cite{GePo23}, Gonczarowski and Moses~\cite{GoMo24}). The critique that in real-life scenarios, we are often closer to something like ``almost'' common knowledge (Halpern~\cite{Ha86}, Halpern and Moses~\cite{HaMo90}, Fagin et al.~\cite{FaHaMoVa04}, Rubinstein~\cite{Ru89}) has led theorists to investigate notions of \emph{common belief}. Monderer and Samet~\cite{MoSa89} show that under \emph{common $p$-belief}, an approximate agreement result holds. Geanakoplos~\cite{Ge94} and Morris~\cite{Mo99} extend this result to relaxed versions of common $p$-belief. These accounts, however, do not cover a different line of generalization of Aumann's theorem that consists in moving from countable to general probability spaces and from knowledge of an \emph{event} to knowledge of a \emph{random variable} (Nielsen~\cite{Ni84}). The present article brings these two strands together by introducing two novel notions of approximate common knowledge, formulated in the language of $\sigma$-algebras and applicable to general probability spaces.  
		 
The first proposed concept, \emph{$(\delta, \varepsilon)$-common knowledge of an event}, relies on two relaxations of elementary set-theoretic notions: (1) an event $B$ being \emph{$\delta$-nearly} contained in a $\sigma$-algebra, and (2) an event $B$ being \emph{$\varepsilon$-nearly} contained in another event $A$. We establish the following results:

\begin{itemize}
	
	\item [(1)] Equivalence of the $\sigma$-algebra-based definition of $(\delta, \varepsilon)$-common knowledge of an event 
    to both a \emph{hierarchical} and an \emph{alternating hierarchical} definition thereof, generalizing Aumann's~\cite{Au76} argument showing the equivalence of the partition-based definition and the informal alternating hierarchical definition of common knowledge (Proposition \ref{prop:equivalences_of_near_common_knowledge}). 
	
	\item [(2)] A generalization of Aumann's~\cite{Au76} Agreement Theorem, showing that when two individuals have $(\delta, \varepsilon)$-common knowledge of their posteriors of an event---or more generally, of a random variable taking values in the unit interval---then the distance between these posteriors is bounded as a function of $\delta$ and  $\varepsilon$ (Theorem \ref{theo:aumann_for_delta_eps_ck_schwach} and Theorem \ref{theo:aumann_for_delta_eps_ck}).
\end{itemize}
These results are related to those of Geanakoplos~\cite{Ge94} and Morris~\cite{Mo99} for weak common $p$-belief, which have been established for countable probability spaces. In contrast, all results of the present paper hold for general probability spaces. Section~3.3 provides a more detailed discussion of the relationship between $(\delta, \varepsilon)$-common and weak common $p$-belief, and shows in particular how one notion can be converted into the other (Proposition \ref{lemma:delta_to_p} and Proposition \ref{lemma:p_to_delta}). 

For the second proposed concept, which is central to this work, we extend Nielsen's~\cite{Ni84} notion of common knowledge of a \emph{random variable} to a notion of \emph{$(\delta, \varepsilon)$-common knowledge of a random variable}, formulated in terms of \emph{conditional variances} (Definition~\ref{def-X-deCK}). In this setting, we establish the following results: 

\begin{itemize}
	\item [(3)] Reformulations of $(\delta, \varepsilon)$-common knowledge of a random variable with both a \emph{hierarchical} and an \emph{alternating hierarchical} definition of \emph{$(\delta, \varepsilon)$-common knowledge of a random variable} (Proposition \ref{prop:equivalences_of_near_common_knowledge_X}).
	
	\item [(4)] An agreement theorem for $(\delta, \varepsilon)$-common knowledge of a \emph{random variable} which states that if the posteriors of a random variable $X$ are $(\delta,\epsilon)$-common knowledge on an event $B$, then the $L_2$-distance between the posteriors on $B$ is bounded as a function of $\delta$ and $\epsilon$ (Theorem \ref{theo:localaumann}). 
	
	\item [(5)] Finally, in extension of 
    Geanakoplos and Polemarchakis~\cite{GePo82} and Nielsen~\cite{Ni84}, a dynamic approximate agreement result (Theorem~\ref{theo:dialoge_epsilon}), which shows that, when individuals keep learning information of a random variable $X$, it suffices to merely approach $\epsilon$-common information (Definition~\ref{def-X-CK}) to guarantee that $\epsilon$-common information of $X$ is reached at infinity. We illustrate this by a scenario of communication with noise (Example \ref{ex:com_noise}).
	
\end{itemize}

The versatility of $(\delta, \varepsilon$)-common knowledge rests on a conceptual generalization: instead of defining $(\delta, \varepsilon)$-knowledge and common knowledge of an event \emph{at a particular state} $\omega$, we define it \emph{on an event}. 
We begin in Section 2 by motivating this approach and reformulating Aumann's~\cite{Au76} results within this framework.  

\section{Preliminaries: knowledge ``on'' an event}

Aumann uses partitions to model individuals' knowledge and common knowledge and states his---now iconic---theorem in that language.  As shown by Nielsen~\cite{Ni84}, the result can be generalized to a framework where individuals' knowledge is modeled by $\sigma$-algebras. In this more general framework, Nielsen extends the notion of common knowledge of an \emph{event} to common knowledge of a \emph{random variable} and shows an agreement result for knowing a random variable. We follow Nielsen in this approach, but generalize it one step further to accommodate nearby knowledge. 

Nielsen utilizes the notion (and existence) of the largest set on which an individual knows an event. For our relaxation of knowledge and common knowledge, this approach is not viable, as there are typically many different events of the same size that provide common knowledge of an event up to a given margin. To address this, we define what it means for an individual to know an event $A$ \emph{on an event} $B$.

\subsection{The basic setup}

We fix an ambient probability space $(\Omega, \F, \P)$, where $\Omega$ represents all possible states of the world. Individuals who reason about the world are each identified with a $\sigma$-algebra representing their knowledge of the world. An event is in an individual's $\sigma$-algebra precisely if the individual can discern whether that event occurs. 

We assume all $\sigma$-algebras to be complete, that is, every subset of every null set is an element of the $\sigma$-algebra. Furthermore, we treat sets that differ only on null sets as equivalent.

To avoid notational complexity, all definitions and results are stated for two individuals. The main results can all be reformulated to allow for more than two individuals and can be shown with virtually the same proofs as those presented here.

\begin{definition}\label{def:knowledge_on_B}
	Let $\G$ be a $\sigma$-algebra  and $A, B$  events, $\P(B)>0$. Then, we say that \emph{$\G$ knows $A$ on $B$} if $B \in \G$ and $B \subseteq A$. This property is denoted by $K(\G,A,B)$.
\end{definition}
With this, we can also define common knowledge on an event.

\begin{definition}\label{def:cn}

	Let $\G_1, \G_2$ be $\sigma$-algebras and $A, B$ events, $\P(B)>0$. Then, \emph{$A$ is common knowledge on $B$} if $B \in \G_1\cap \G_2$ and $B \subseteq A$. This property is denoted by $K(A,B).$ 

\end{definition}

\begin{remark}
	In the case that $\Omega$ satisfies $\P(\{ \omega \})>0$ for all $\omega \in \Omega$, this notion of common knowledge is essentially equivalent to Aumann's notion. If $A$ is common knowledge on an event $B$, then $A$ is common knowledge at every $\omega \in B$. Conversely, if $A$ is common knowledge on some $\omega$, then $A$ is common knowledge on $B$, where $B$ is the smallest event in $\G_1 \cap \G_2$ that contains $\omega$. 
\end{remark}

As is well known, it is possible to give an alternative characterization of $K(A,B)$ based on the informal notion originally suggested by Lewis~\cite{Le69}, according to which an event $A$ is common knowledge, in a group of people, if everyone knows that it occurs, everyone knows that everyone knows that it occurs, and so on.\footnote{See Lewis~\cite{Le69}, p. 52 and following.} 
 The following definition, which we refer to as the \emph{hierarchical} definition of common knowledge, formalizes this verbal description.

\begin{definition}\label{def-hierarchical-Lewis}
	Let $\G_1,\G_2$ be $\sigma$-algebras and $A,B$ events, $\P(B)>0$. Then, \emph{$A$ is hierarchically common knowledge on $B$} if there exist sequences $(C_n)_{n\in \N}$ and $(D_n)_{n\in \N}$ such that the following conditions hold for all $n \in \N$:
    \begin{itemize}
        \item $K(\G_1, A, C_n)$ and $K(\G_2,A,D_n)$,   
        \item $C_{n+1} \subseteq C_n \cap D_n$ and $D_{n+1} \subseteq C_n \cap D_n$, 
        \item $B=\bigcap_{n\in \N} \left( C_n \cap D_n \right)$.
    \end{itemize}
 
\end{definition}

\begin{remark}
       In the above, $C_1$ is an event on which individual 1 knows $A$; $D_1$ an event on which individual 2 knows $A$; $C_2$ an event on which individual 1 knows that both individuals know  $A$; $D_2$ an event on which individual 2 knows that both individuals know $A$; and so on.
\end{remark}

\begin{remark}
	Note that $\bigcap_{n\in \N} \left( C_n \cap D_n \right) = \bigcap_{n\in \N} C_n = \bigcap_{n\in \N} D_n $ and thus the condition  $B=\bigcap_{n\in \N} \left( C_n \cap D_n \right)$ could be equivalently posed as  $B=\bigcap_{n\in \N} C_n $ or  $B=\bigcap_{n\in \N} D_n$.
\end{remark}

For some generalizations of common knowledge it is useful to consider a variant of the hierarchical definition encapsulating the idea that an event is common knowledge between two individuals if individual $1$ knows that it occurs, individual $2$ knows that individual $1$ knows that it occurs, individual $1$ knows that individual $2$ knows that individual $1$ knows it occurs, and so on ad infinitum. This is the verbal description given by Aumann.

To distinguish it from the notion given in Definition \ref{def-hierarchical-Lewis}, we refer to it as the 
\emph{alternating hierarchical} definition of common knowledge.

\begin{definition}
	Let $\G_1,\G_2$ be $\sigma$-algebras and $A,B$ events, $\P(B)>0$. Then, \emph{$A$ is alternatingly hierarchically common knowledge on $B$} if there exists a sequence of events $(B_n)_{n\in \N}$ such that $B=\bigcap_{n\in \N} B_n$, $B_i \subseteq B_j$ for $j < i$, and $K(\G_1, A, B_n)$ if $n$ is odd and $K(\G_2,A,B_n)$ if $n$ is even.    
\end{definition} 

\begin{prop}(Aumann, Nielsen)\label{lemma:ck_in_two_ways}
	Let $\G_1,\G_2$ be $\sigma$-algebras, and $A,B$ events, $\P(B)>0$. Then, the following conditions are equivalent:
	\begin{itemize}
		\item [(1)]  $A$ is common knowledge on $B$.
		\item [(2)]  $A$ is hierarchically common knowledge on $B$.
		\item [(3)]  $A$ is alternatingly hierarchically common knowledge on $B$.
	\end{itemize}
\end{prop}

The proof of this equivalence is provided in a more general setting in Proposition~\ref{prop:equivalences_of_near_common_knowledge}.

In the language ``on $B$,'' Aumann's agreement theorem can be restated as follows.

\begin{theorem}{(Aumann)} \label{theo:Aumann}
	Let $\G_1,\G_2$ be $\sigma$-algebras, $A$ an event, and $q_1,q_2 \in [0,1]$. If the events $\{\P(A|\G_1)=q_1\}$, $\{ \P(A|\G_2)=q_2\}$ are common knowledge on some event $B$, $\P(B)>0$, then $q_1=q_2$.
\end{theorem}

These definitions and results will serve as templates for the approximate versions of knowledge and common knowledge introduced in Sections 3 and 4.

\boldmath 
\section{ ${(\delta,\epsilon)}$-common knowledge of an event and agreement}
\unboldmath 

  This section introduces \emph{$(\delta,\epsilon)$-common knowledge} of an \emph{event} defined on the basis of $\sigma$-algebras, shows that this notion is equivalent to both a hierarchical and an alternating hierarchical definition, and extends Aumann's~\cite{Au76} result to that framework.

\boldmath
\subsection{$(\delta,\epsilon)$-common knowledge on an event}
\unboldmath
 
We begin by introducing a notion that captures that an event is almost an element of a $\sigma$-algebra.

\begin{definition}
	Let $\G$ be a $\sigma$-algebra and $B$ an event. Then, \emph{$B$ is $\delta$-nearly in $\G$} if there exists an event $G \in \G$ such that  $$ \P(B \triangle G)\leq  \P(B) \delta,$$
	where the symmetric difference between $B$ and $G$ is defined as
	$$B\triangle G := (B\setminus G) \cup (G \setminus B) .$$
	For this property, we write $B \in \G^\delta$.
\end{definition}

Note that  $A \in \G^0$ is equivalent to $A \in \G$ for any complete $\sigma$-algebra $\G$. Next, we define what it means for an event to be nearly contained in another event.

\begin{definition}	
	Let $A,B$ be events. Then, \emph{$B$ is $\epsilon$-nearly a subset of $A$} if $$\P(B\setminus A) \leq \P(B)\epsilon.$$
	For this property, we write $B \subseteq_\epsilon A$.
\end{definition}

Note that $\subseteq_\epsilon$ is not a transitive relation.

Given an event $A$ and a $\sigma$-algebra $\G$, for what follows, it is helpful to know if there is a smallest $\delta$ such that $A \in \G^\delta$. The following lemma establishes that such a $\delta$ always exists by a constructive argument explicitly determining it.

\begin{lemma}\label{lemma:smallest_delta}
	Let $\G$ be a  $\sigma$-algebra and $A$ an event. Define $G\in \G$ as $G:=\{\P(A|\G) \geq \frac{1}{2}\}$. Then, 
	$$\P(A \triangle G) = \inf_{\tilde{G} \in \G} \P(A \triangle \tilde{G} ).$$
	 
	 In particular, $\P(A \triangle G)/\P(A)$ is the smallest element of the set $\{\delta : A \in \G^\delta\}$.
\end{lemma}

\begin{proof}
We need to show that $\P(A \Delta G) \le \P(A \Delta \tilde{G})$ for every $\tilde{G} \in \G$. To that end, fix $\tilde{G} \in \G$. For any events $B,C$, the inequality $\P(B) \leq \P(C)$ is equivalent to $\P(B\setminus C) \leq \P(C\setminus B)$. Hence, $\P(A\triangle G) \leq \P(A \triangle \tilde{G})$
	holds if and only if the inequality
	\begin{align*}
		\P(A \cap (\tilde{G}\setminus G)) + \P(A^\complement \cap (G\setminus \tilde{G})) =	 \P((A \triangle G) \setminus(A\triangle \tilde{G}) ) \\
		 \leq \P((A \triangle \tilde{G}) \setminus(A\triangle G) ) = \P(A \cap (G\setminus \tilde{G})) + \P(A^\complement \cap (\tilde{G}\setminus G))
	\end{align*}
	holds.
	To prove this inequality, we show that 
	\begin{align}\label{align:minimal_delta_1}
		\P(A \cap (\tilde{G}\setminus G)) \leq   \P(A^\complement \cap (\tilde{G}\setminus G)) 
	\end{align}
	as well as 
	\begin{align}  \label{align:minimal_delta_2}
		\P(A^\complement \cap (G\setminus \tilde{G})) \leq  \P(A \cap (G\setminus \tilde{G})). \end{align}
	To prove inequality~\eqref{align:minimal_delta_1}, we calculate
	\begin{align}\label{align:minimal_delta_3}
	\P(A \cap (\tilde{G} \setminus G)| \G)= \ind{\tilde{G} \setminus G}  \P(A|\G) \leq \frac{1}{2} \ind{\tilde{G} \setminus G}  \leq \ind{\tilde{G} \setminus G} \P(A^\complement|\G) = \P(A^\complement \cap (\tilde{G} \setminus G)|\G), \end{align}
	where the inequalities follow directly from the definition of $G$. Now, taking expectations over inequality~\eqref{align:minimal_delta_3} proves inequality~\eqref{align:minimal_delta_1}. Inequality~\eqref{align:minimal_delta_2} can be shown by a similar argument. 
\end{proof}

After this preparation, we can formulate relaxed notions of knowledge and common knowledge.
\begin{definition}
	Let $\G$ be a $\sigma$-algebra, $A,B$ events, $\P(B)>0$, and $\delta,\epsilon \geq 0$. We say that \emph{$A$ is  $(\delta,\epsilon)$-known on $B$ by $\G$} if $B\in \G^\delta$ and $B \subseteq_\epsilon A$. We denote this property by $K_{\delta,\epsilon}(\G, A,B).$
\end{definition}

\begin{definition}
	Let $\G_1, \G_2$ be $\sigma$-algebras, $A, B$ events, $\P(B)>0$, and $\delta,\epsilon \geq 0$. Then, \emph{$A$ is $(\delta, \epsilon)$-common knowledge on $B$} if $B\in \G_1^\delta \cap \G_2^\delta$ and $B \subseteq_\epsilon A.$  We denote this property by $K_{\delta,\epsilon}( A,B)$.  
\end{definition}

\begin{remark}
	If $A$ is (common) knowledge on $B$, then $A$ is also $(0,0)$-(common) knowledge. For the converse, if $A$ is $(0,0)$-(common) knowledge on $B$, then it is also (common) knowledge on $B$. 
\end{remark}

\begin{remark}\label{remark: self-evident} An event $A$ is $(\delta, \epsilon)$-common knowledge on $B$ if and only if $K_{\delta,\epsilon}(\G_1, A,B)$ and $K_{\delta,\epsilon}(\G_2, A,B)$, which splits the notion of $(\delta, \epsilon)$-common knowledge into two conditions, each referring to only one individual at a time. This reflects a well-known property of classical common knowledge, when looking at it through the lens of \emph{self-evident events}: 
  An event $A$ is common knowledge at $\omega$ if there exists an event $B$, so that $\omega \in B$ and $B$ is self-evident to both individuals, that is, $B \in \G_i$ for both $i=1$ and $i=2$ (see, for instance, Geanakoplos \cite{Ge94}).
    \end{remark}

Next, we confirm that $(\delta,\epsilon)$-common knowledge can be equivalently formulated by both a hierarchical and an alternating hierarchical definition.

\begin{definition}
	Let $\G_1,\G_2$ be $\sigma$-algebras, $A,B$ events, $\P(B)>0$, and $\delta,\epsilon \geq 0$. Then, \emph{$A$ is hierarchically $(\delta, \epsilon)$-common knowledge on $B$} if there exist sequences $(C_n)_{n\in \N}$ and $(D_n)_{n\in \N}$ such that the following conditions hold for all $n \in \N$:
    \begin{itemize}
        \item $K_{\delta,\epsilon}(\G_1, A, C_n)$ and $K_{\delta,\epsilon}(\G_2,A,D_n)$,  
        \item $C_{n+1} \subseteq C_n \cap D_n$ and $D_{n+1} \subseteq C_n \cap D_n$, 
        \item $B=\bigcap_{n\in \N} \left( C_n \cap D_n \right)$.
    \end{itemize}
\end{definition}

\begin{definition}
	Let $\G_1,\G_2$ be $\sigma$-algebras, $A,B$ events, $\P(B)>0$, and $\delta,\epsilon \geq 0$. Then,  \emph{$A$ is alternatingly hierarchically $(\delta,\epsilon)$-common knowledge on $B$} if there exists a sequence of events $(B_n)_{n\in \N}$ such that $B=\bigcap_{n\in \N} B_n$,  $B_i \subseteq B_j$ for $j < i$, and $K_{\delta,\epsilon}(\G_1, A,B_n)$ if $n$ is odd and  $K_{\delta,\epsilon}(\G_2, A,B_n)$ if $n$ is even.
\end{definition}

To show that these notions of $(\delta,\epsilon)$-common knowledge agree, we start with a preparatory lemma.

\begin{lemma}\label{lemma:falling_sequence}
	Let $\G$ be a $\sigma$-algebra, $A$ an event, $(B_i)_{i\in \N}$ a sequence of events such that $B_i \subseteq B_j$ for $j<i$, and $\delta, \epsilon \geq 0.$ If for all $i\in \N$,
	$$ K_{\delta,\epsilon}(\G, A,B_i),$$
	then
	$$K_{\delta,\epsilon}\left(\G, A, \bigcap_{i\in \N} B_i\right).$$
\end{lemma}

\begin{proof}
	Let $B:= \bigcap_{i\in \N}B_i.$ We have to show that $B \in \G^\delta$ and $B \subseteq_\epsilon A.$
	
	We start with $B \in \G^\delta$. Let $\tilde{\epsilon} >0$ be given. As probability measures are continuous from above, we can find $N \in \N$, such that
	$$\P(B_n \triangle B) = \P(B_n)- \P(B) \leq \tilde{\epsilon},$$
	for all $n > N$.
	In particular, $\P(B_n) \leq \tilde{\epsilon} + \P(B)$ for such $n$.
	
	As $B_i \in \G^\delta$ for $i \in \N$, there exists $G_i \in \G$ such that $\P(B_i \triangle G_i) \leq \delta \P(B_i).$
	Using the triangle inequality, we find that for all $n > N$,
	\begin{align*}
		\P(B \triangle G_n) &\leq \P(B \triangle B_n) + \P(B_n \triangle G_n) \leq \tilde{\epsilon} + \delta \P(B_n) \\
		&\leq \tilde{\epsilon} + \delta (\tilde{\epsilon} + \P(B))= \delta \P(B) + \tilde{\epsilon}(1+\delta).
	\end{align*}
	As $\tilde{\epsilon}$ was arbitrarily small, we conclude $B \in \G^{\tilde{\delta}}$ for all $\tilde{\delta}> \delta$. Lemma~\ref{lemma:smallest_delta} then implies that also $B \in \G^\delta$.
	
	To show that $B \subseteq_\epsilon A$, we simply observe that for all $i \in \N$,
	$$
	\P(B_n \setminus A) \leq \epsilon \P(B_n), 
	$$
	and hence, by taking limits and by the continuity of probability measures from above, 
	\[\P(B\setminus A) \leq \epsilon \P(B). \qedhere \]
\end{proof}

\begin{prop}\label{prop:equivalences_of_near_common_knowledge}
	Let $\G_1,\G_2$ be $\sigma$-algebras, $A,B$ events, and $\delta, \epsilon \geq 0$. Then, the following conditions are equivalent:
	\begin{itemize}
		\item [(1)]  $A$ is $(\delta, \epsilon)$-common knowledge on $B$. 
		\item [(2)]  $A$ is hierarchically $(\delta, \epsilon)$-common knowledge on $B$. 
		\item [(3)]  $A$ is alternatingly hierarchically $(\delta, \epsilon)$-common knowledge on $B$.
	\end{itemize}
\end{prop}

\begin{proof}
	To show that $A$ being $(\delta,\epsilon)$-common knowledge on $B$ implies $A$ being hierarchically  $(\delta,\epsilon)$-common knowledge on $B$, simply set $C_n:=D_n:=B$ for all $n \in \N$.
	
	To show that $A$ being hierarchically $(\delta,\epsilon)$-common knowledge on $B$ implies $A$ being alternatingly hierarchically  $(\delta,\epsilon)$-common knowledge on $B$, let $(C_n)_{n\in \N}$ and $(D_n)_{n\in \N}$ be given as in the definition of hierarchical $(\delta,\epsilon)$-common knowledge. Now, define the sequence of events $(B_n)_{n \in \N}$ as $B_n:=C_n$ for odd $n$ and $B_n:=D_n$ for even $n$. It is now straightforward to check that the sequence $(B_n)_{n \in \N}$ shows that $A$ is alternatingly hierarchically $(\delta,\epsilon)$-common knowledge on $B$.
	
	Lastly, we have to check that alternating hierarchical $(\delta,\epsilon)$-common knowledge implies $(\delta,\epsilon)$-common knowledge. For this, let $(B_n)_{n \in \N}$ be a sequence of events as in the definition of alternating hierarchical $(\delta,\epsilon)$-common knowledge. We introduce the monotonically decreasing sequence of events $C_n:=B_{2n+1}$ for, $n \in \N$, and note that $B=\bigcap_{i \in \N}C_i.$ By assumption, $K_{\delta,\epsilon}(\G_1,A,C_n)$ for all $n \in \N$, and hence, by Lemma~\ref{lemma:falling_sequence}, $K_{\delta,\epsilon}(\G_1,A,B)$. The property $K_{\delta,\epsilon}(\G_2,A,B)$ is proven similarly, using the sequence $(D_n)_{n\in \N}$ defined by $D_n:=B_{2n}$ instead of the sequence $(C_n)_{n \in \N}.$
\end{proof}

Note that we obtain Proposition~\ref{lemma:ck_in_two_ways} by setting $\delta=\epsilon=0$ in Proposition~\ref{prop:equivalences_of_near_common_knowledge}.

\boldmath
\subsection{ 
Aumann's agreement theorem under $(\delta,\epsilon)$-common knowledge}
\unboldmath

The aim of this section is to establish results that show the robustness of Aumann's Agreement Theorem. Our results show that even when common knowledge is weakened, substantial disagreement about the likelihood of an event is impossible.

Specifically, Theorem~\ref{theo:aumann_for_delta_eps_ck_schwach} establishes the following: 
If it is $(\delta,\epsilon)$-common knowledge that the two individuals' posterior probabilities of a given event lie in some intervals $[a_1,b_1]$ and $[a_2,b_2]$, respectively, then the difference between their posteriors is uniformly bounded in terms of $\delta$, $\epsilon$ and the widths of the intervals.

\begin{theorem}\label{theo:aumann_for_delta_eps_ck_schwach}
	Let $\G_1,\G_2$ be $\sigma$-algebras  and $C$ an \emph{event}.	
	Define 
	$$A:=\{\P(C|\G_1)\in [a_1,b_1] \} \cap  \{\P(C|\G_2)\in [a_2,b_2] \}$$
	for some $ a_1,b_1,a_2,b_2 \in [0,1]$.	
	If there exists an event $B$, $\P(B)>0$, and $\delta, \epsilon \in [0,1]$, $3\delta + 2\epsilon \leq 1$, such that $A$ is $(\delta,\epsilon)$-common knowledge on $B$, then there exist $q_1 \in [a_1,b_1]$ and $q_2 \in [a_2,b_2]$ such that 	
	\begin{align}\label{align:Set_Aumann_1}
		|q_1-q_2| \leq \frac{2\delta+ \epsilon}{1+\delta} \leq 2(\delta+\epsilon).
	\end{align}
	In particular, on $A$ we obtain 
	$$|\P(C|\G_1)-\P(C|\G_2)| \leq b_1-a_1+b_2-a_2+2(\delta+\epsilon).$$ 
	
\end{theorem}
\begin{proof}
Theorem~\ref{theo:aumann_for_delta_eps_ck_schwach} is an easy consequence of the more general Theorem~\ref{theo:aumann_for_delta_eps_ck} below.     
\end{proof}

As the choice of $a_1=b_1$ and $a_2=b_2$ is admissible in Theorem~\ref{theo:aumann_for_delta_eps_ck_schwach}, considering events of the type $\{\P(C|\G) \in [a,b]\}$ is a generalization of the usual Aumann result which concerns events of the form $\{\P(C|\G)=q\}$. The motivation for this is twofold. First, in continuous models, events of the form $\{\P(C|\G)=q\}$ are typically null sets. Hence, such events are usually not $(\delta,\epsilon)$-common knowledge on any set $B$ with $\P(B)>0$, and Aumann-type theorems based on such conditions are therefore typically not meaningful in continuous settings.

Second, for $\delta=\epsilon=0$, Theorem~\ref{theo:aumann_for_delta_eps_ck_schwach} is in its own right an extension of Aumann’s agreement theorem. It covers the situation where the precise values of the posteriors of individuals 1 and 2 are not common knowledge, but it is however common knowledge that they are within certain ranges $[a_1,b_1]$ and $[a_2,b_2]$. In this case these ranges necessarily overlap, i.e. the individuals cannot disagree entirely. The following Corollary states this explicitly.

\begin{corollary}
    Let $\G_1,\G_2$ be $\sigma$-algebras and $C$ be an event.
    Define
    $$A:= \{\P (C | \G_1) \in [a_1,b_1] \} \cap \{ \P( C| \G_2) \in [a_2,b_2] \},$$
    for some $a_1,b_1,a_2,b_2 \in [0,1]$.
    If there exists an event $B$, $\P(B)>0$, such that $A$ is common knowledge on $B$, then 
    $[a_1,b_1] \cap [a_2,b_2] \neq \emptyset. $
    
    In particular, on $A$ we obtain
    $$|\P(C | \G_1) - \P(C|\G_2)| \leq b_1 -a_1 + b_2 +a_2.$$
\end{corollary}
\begin{proof}
This is an immediate consequence of  Theorem~\ref{theo:aumann_for_delta_eps_ck_schwach} with $\delta=\epsilon=0$.
\end{proof}

\begin{remark}
	
	For \emph{countable} probability spaces, results similar to Theorem~\ref{theo:aumann_for_delta_eps_ck_schwach} have been shown by Monderer and Samet \cite{MoSa89} in terms of \emph{common $p$-belief} (bounds for which were improved by Neeman~\cite{Ne96}), and by Geanakoplos \cite{Ge94} and Morris \cite{Mo99} in terms of \emph{weak common $p$-belief}. Our results apply to \emph{general} probability spaces. A comparison of weak common $p$-belief to $(\delta,\epsilon)$-common knowledge of an event is provided in Section \ref{Sec-comparision}.
\end{remark}

Theorem~\ref{theo:aumann_for_delta_eps_ck_schwach}, which, as Aumann's theorem, is about the individuals' posteriors attributed to an \emph{event}, is a special case of the following more general agreement theorem about the individuals' posterior expectation of a \emph{random variable}.

\begin{theorem}\label{theo:aumann_for_delta_eps_ck}
	Let $\G_1,\G_2$ be $\sigma$-algebras  and $X$  a random variable taking values in $[0,1]$.
	Define 
	$$A:=\{\E[X|\G_1]\in [a_1,b_1] \} \cap  \{\E[X|\G_2]\in [a_2,b_2] \}$$
	for some real numbers $ a_1,b_1,a_2,b_2 \in [0,1]$.
	If there exists an event $B$, $\P(B)>0$, and $\delta, \epsilon \in [0,1]$, $3\delta + 2\epsilon \leq 1$,  such that $A$ is $(\delta,\epsilon)$-common knowledge on $B$, then there exist $q_1 \in [a_1,b_1]$ and $q_2 \in [a_2,b_2]$  such that 
	\begin{align}\label{align:Set_Aumann_1}
		|q_1-q_2| \leq \frac{2\delta+ \epsilon}{1+\delta} \leq 2 (\delta+\epsilon).
	\end{align}	
\end{theorem}

We give the proof for the more general Theorem~\ref{theo:aumann_for_delta_eps_ck}. For this, it is useful to isolate the following lemma, which is also an important step in the proof of Neeman's~\cite{Ne96} Aumann-type result.

\begin{lemma}\label{lemma:ungl_fuer_aumann}
	Let $Z$ be a random variable taking values in $[0,1]$ and $E,F$ events such that $\P(E\cap F) > 0$. Then,
	\begin{equation*}
		\frac{\E[Z \ind{E}]}{\P(E)} \geq \P(F|E) \frac{\E[Z \ind{E\cap F}]}{\P(E\cap F)}.
	\end{equation*}
\end{lemma}

\begin{proof}
	We calculate
	\begin{equation*}
		\frac{\E[Z \ind{E}]}{\P(E)} \geq \frac{\E[Z \ind{E\cap F}]}{\P(E)} = \frac{\P(E\cap F)}{\P(E)} \frac{\E[Z \ind{E\cap F}]}{\P(E\cap F)}=\P(F|E) \frac{\E[Z \ind{E \cap F}]}{\P(E\cap F)}. \qedhere
	\end{equation*}
\end{proof}

\begin{proof}[Proof of Theorem~\ref{theo:aumann_for_delta_eps_ck}]
	For $i\in \{1,2\}$, let $A_i:= \{\P(A|\G_i) > 0\} \in \G_i$, $B_i \in \G_i$ such that $\P(B \triangle B_i) \leq \delta \P(B)$, and $ C_i:=A_i\cap B_i \in \G_i$.

	We want to apply Lemma~\ref{lemma:ungl_fuer_aumann}  to $X$ and the sets $E=C_i$ and $F=C_j$ for $i,j\in \{1,2\}$. We first give a lower bound for $\P(C_i|C_j)$:
	\begin{align*}
		\P(C_i|C_j) = & \frac{\P(C_i \cap C_j)}{\P(C_j)} = \frac{\P(A_i \cap B_i \cap A_j \cap B_j)}{\P(A_j \cap B_j)} \geq \frac{\P(B_i \cap B_j \cap A)}{\P(A_j \cap B_j)}\\
		&\geq \frac{\P(B_i\cap B_j \cap A)}{\P(B_j)} \geq  \frac{\P(B_i\cap B_j \cap A \cap B)}{\P(B_j)} \\ 
		&= \frac{\P(B_i\cap B_j \cap B) - \P(B_i \cap B_j \cap B \cap A^{\complement} )}{\P(B_j)} \\
		& \geq \frac{\P(B) - \P(B\setminus B_i) - \P(B\setminus B_j)- \P(B_i \cap B_j \cap B \cap A^{\complement} )}{\P(B_j)} \\
		& \geq \frac{\P(B) - \P(B\setminus B_i) - \P(B\setminus B_j)- \P( B \cap A^{\complement})}{\P(B_j)} \\
		&\geq \frac{\P(B) (1-\delta-\epsilon) - \P(B\setminus B_j)}{\P(B_j)}.
	\end{align*}
As $\P(B \triangle B_j) \leq \P(B) \delta$, we can  write $\P(B\setminus B_j) = a \P(B)$ and $\P(B_j\setminus B)=b \P(B)$ with $a+b \leq \delta$, and so  $\P(B_j)=\P(B) - \P(B\setminus B_j) + \P(B_j \setminus B) = \P(B) (1-a+b)$. An easy computation shows that, under the condition that $3\delta + 2 \epsilon \leq 1$,
	$$\frac{\P(B) (1-\delta-\epsilon) - \P(B\setminus B_j)}{\P(B_j)}=\frac{\P(B) (1-\delta-\epsilon-a)}{\P(B)(1-a+b)} =\frac{ (1-\delta-\epsilon-a)}{(1-a+b)}$$
	is minimal for $a=0, b=\delta$. Hence,

	\begin{align}\label{align:Set_Aumann_2}
	\P(C_i|C_j) \geq \frac{ 1-\delta-\epsilon }{1+\delta}.
	\end{align} 
	Note that this also implies that $\P(C_1 \cap C_2)>0$.
	Now, we confirm that $C_i \subseteq \{\E[X|\G_i] \in [a_i,b_i] \}$. As $A \subseteq  \{\E[X|\G_i] \in [a_i,b_i] \}$,
	\begin{align}\label{align:Set_Aumann_2.5}
	 \P(A|\G_i) \leq \P( \{\E[X|\G_i] \in [a_i,b_i] \}|\G_i) = \ind{ \{\E[X|\G_i] \in [a_i,b_i] \}},
	 \end{align}
	using that $ \{\E[X|\G_i] \in [a_i,b_i] \} \in \G_i$. On $C_i$, we know that $\P(A|\G_i) > 0$, inequality~\eqref{align:Set_Aumann_2.5} hence implies that $ \ind{\{\E[X|\G_i] \in [a_i,b_i] \}}=1$ on $C_i$, that is, $ \E[X|\G_i] \in [a_i,b_i]$ on $C_i$.
	
	Next, we calculate
	\begin{align}\label{align:Set_Aumann_3}
		q_i:= \frac{\E[X\ind{C_i}]}{\P(C_i)} = \frac{\E[\E[X\ind{C_i}|\G_i]]}{\P(C_i)} =  \frac{\E[\ind{C_i}\E[X|\G_i]]}{\P(C_i)} \in [a_i,b_i], 
	\end{align}
	where we used that $C_i \in \G_i$, the tower property,\footnote{If 
		$\mathcal{H}$ is a sub $\sigma$-algebra of $\G$, then $\E [\E [X \mid \G]\mid \mathcal{H} ] = \E [X \mid \mathcal{H}]$, see, for instance, Williams \cite[p.\ 88]{Wi91}.}  and that $\E[X|\G_i] \in [a_i,b_i]$  on $C_i$.
	
	Lemma~\ref{lemma:ungl_fuer_aumann}, together with \eqref{align:Set_Aumann_2} and \eqref{align:Set_Aumann_3}, yields
	$$q_i \geq \frac{1-\delta-\epsilon}{1+\delta} \frac{\E[X \ind{ C_1\cap C_2}]}{\P(C_1\cap C_2)}.$$
	Applying Lemma~\ref{lemma:ungl_fuer_aumann} now to $1-X$ and keeping $E=C_i$ and $F=C_j$, we get 
	\begin{align*}
	1-q_i  &=	\frac{\E[(1-X) \ind{C_i}]}{\P(C_i)} \geq \P(C_i|C_j) \frac{\E[(1-X) \ind{ C_1\cap C_2}]}{\P(C_1\cap C_2)} \\
	&\geq  \frac{1-\delta-\epsilon}{1+\delta}\frac{\E[(1-X) \ind{ C_1\cap C_2}]}{\P(C_1\cap C_2)} = \frac{1-\delta-\epsilon}{1+\delta} \left(1- \frac{\E[X \ind{ C_1\cap C_2}]}{\P(C_1\cap C_2)}\right)
	\end{align*} 
and hence an upper bound for $q_i$,
\begin{align*}
	&q_i \leq  1 - \frac{1-\delta-\epsilon}{1+\delta} +  \frac{1-\delta-\epsilon}{1+\delta} \frac{\E[X \ind{ C_1\cap C_2}]}{\P(C_1\cap C_2)} \\
	 &= \frac{1+\delta}{1+\delta} - \frac{1-\delta-\epsilon}{1+\delta} +  \frac{1-\delta-\epsilon}{1+\delta} \frac{\E[X \ind{ C_1\cap C_2}]}{\P(C_1\cap C_2)} = \frac{2\delta+\epsilon}{1+\delta} + \frac{1-\delta-\epsilon}{1+\delta} \frac{\E[X \ind{ C_1\cap C_2}]}{\P(C_1\cap C_2)}.
\end{align*}	
	We see that both the upper and the lower bounds are independent of $i$, and thus hold for both $q_1$ and $q_2$, which means that
	\[|q_1-q_2| \leq \frac{2\delta+\epsilon}{1+\delta}.\qedhere\]
\end{proof}

\boldmath
\subsection{Comparison of weak common $p$-belief and $(\delta,\epsilon)$-common knowledge of an  event}\label{Sec-comparision}
\unboldmath

Monderer and Samet~\cite{MoSa89} introduce the concept of \emph{common $p$-belief}, which Geanakoplos~\cite{Ge94} generalizes to \emph{weak $p$-common knowledge}, also referred to as  \emph{weak common $p$-belief}. We state the definition of  weak common $p$-belief here in a slightly adapted version from Morris~\cite[Proposition 9]{Mo99}, in our language ``on $B$.''

\begin{definition}
	Let $\G_1,\G_2$ be $\sigma$-algebras, $A,B$ events, $\P(B)>0$, and $p \in [0,1]$. The event $A$ is \emph{weak common $p$-belief on $B$} if $B$ can be written as $B=B_1\cap B_2$ with $B_i \in \G_i$ and
	\begin{itemize}
		\item [(1)] $\P(B_i|B_j) \geq p$
		\item [(2)] $\P(A|B_i) \geq p$,
	\end{itemize}
	for $i,j \in \{1,2\}$.
\end{definition}

Weak common $p$-belief is closely related to the notion of $(\delta,\epsilon)$-common knowledge of an event. In particular, whenever an event is  $(\delta,\epsilon)$-common knowledge with small $\delta$ and small $\epsilon$, then it is weak common $p$-belief with large $p$, and vice versa, as the following two propositions show.

\begin{prop}\label{lemma:delta_to_p}
	Let $\G_1,\G_2$ be $\sigma$-algebras, $A,B$ events, $\P(B)>0$, and $\delta,\epsilon \geq 0$ such that $\delta \leq 1/3$ and $2\epsilon + \delta \leq 1$. If $A$ is $(\delta, \epsilon)$-common knowledge on $B$ with $B_1 \in \G_1, B_2 \in \G_2$ such that $\P(B \triangle B_i) \leq \P(B) \delta$, for $i \in \{1,2\}$, then $A$ is weak common $\frac{1-\max(\delta,\epsilon)}{1+\delta}$-belief on  $B_1 \cap B_2$.
\end{prop}

\begin{proof}
	As $\P(B \triangle B_i) \leq \delta \P(B)$, for $i \in \{1,2\}$, we can define non-negative numbers $a,b,c,d$ by 
	\begin{align*}
	&\P(B_1 \setminus B) = a \P(B),	\\
	&\P(B \setminus B_1) = b \P(B),	\\
	&\P(B_2 \setminus B) = c \P(B),	\\
	&\P(B \setminus B_2) = d \P(B),	
	\end{align*}
	with $a+b \leq \delta$ and $c+d \leq \delta$.
	We first prove $\P(B_2|B_1) \geq (1-\delta)/(1+\delta)$.  The inequality with $B_1$ and $B_2$ swapped then follows in the same way. We calculate
	\begin{align*}
		\P(B_2|B_1) &= \frac{\P(B_1\cap B_2)}{\P(B_1)} \geq \frac{\P(B_1\cap B_2 \cap B)}{\P(B)+\P(B_1 \setminus B)- \P(B \setminus B_1)} \\
		&\geq \frac{\P(B) -  \P(B\setminus B_1) -\P(B \setminus B_2)  }{\P(B)(1+a-b)} = \frac{\P(B) (1-b-d) }{\P(B)(1+a-b)} \\
		&= \frac{1-b-d}{1+a-b}.
	\end{align*}
	An easy calculation shows that $\frac{1-b-d}{1+a-b}$ is minimal for $a=d=\delta, b=c=0$, if $\delta \leq 1/3$. Hence,
	$$	\P(B_2|B_1) \geq \frac{1-\delta}{1+\delta}.$$
	Next, we show that $P(A|B_1) \geq (1-\epsilon)/(1+\delta)$. The inequality $P(A|B_2) \geq (1-\epsilon)/(1+\delta)$ then follows in the same way. We thus calculate 
	\begin{align*}
		\P(A|B_1) &= \frac{P(A\cap B_1)}{\P(B_1)} \geq \frac{\P(B \cap B_1 \cap A)}{\P(B) + \P(B_1 \setminus B) - \P(B \setminus B_1)} \\
		&\geq \frac{\P(B)- \P(B \setminus B_1) - \P(B \setminus A)}{\P(B) + \P(B_1 \setminus B) - \P(B \setminus B_1)} \geq \frac{\P(B)(1-b-\epsilon)}{\P(B)(1+a-b)}=\frac{1-b-\epsilon}{1+a-b}.
	\end{align*}
	Under the assumption that $2\epsilon + \delta \leq 1$, the expression $(1-b-\epsilon)/(1+a-b)$ is minimal for $a=\delta,b=0$, and hence
	\[\P(A|B_1) \geq \frac{1-\epsilon}{1+\delta}. \qedhere\]

\end{proof}

\begin{prop}\label{lemma:p_to_delta}
	Let $\G_1,\G_2$ be $\sigma$-algebras, $A$ an event, $B_1 \in \G_1, B_2 \in \G_2$ events such that $\P(B_1 \cap B_2)>0$, and $p \in [0,1]$. 
	If $A$ is weak common $p$-belief on $B:=B_1\cap B_2$, then $A$ is $(\frac{1-p}{p},\frac{1-p}{p})$-common knowledge on $B$.
\end{prop}

\begin{proof}
	We first show that $$\P(B_1 \triangle B) \leq \P(B)  \frac{1-p}{p}.$$
	The analogous inequality for $B_2$ can be proved in the same way.
	First, observe that the assumption $\P(B_2|B_1)\geq p$ can be rewritten as
	\begin{align}\label{align:p_to_epsilon}
	 \P(B_1) \leq \P(B_1 \cap B_2)/p.
	\end{align}
	 Hence,
	\begin{align*}
		&\P(B \triangle B_1) = \P((B_1 \cap B_2) \triangle B_1)  = \P(B_1) - \P(B_1 \cap B_2) \\
		&\leq \frac{\P(B_1 \cap B_2)}{p} - \P(B_1 \cap B_2) = \P(B_1 \cap B_2) \frac{1-p}{p} =\P(B) \frac{1-p}{p}.
	\end{align*}
	What is left to prove is that $B \subseteq_\epsilon A.$ Note that $\P(A|B_1) \geq p$ is equivalent to $\P(A \cap B_1) \geq p \P(B_1)$, and thus
	\begin{align*}
		\P(B \setminus A) &= \P((B_1 \cap B_2) \setminus A) \leq \P(B_1 \setminus A) = \P(B_1)- \P(A\cap B_1)\\
		&\leq \P(B_1) - p\P(B_1) = \P(B_1) (1-p) \leq \P(B) \frac{1-p}{p},
		\end{align*} 
	 where we used inequality~\eqref{align:p_to_epsilon} in the last step.
\end{proof}

\begin{remark}
	Note that, while each of the two previous lemmas is sharp, alternating applications of these conversions between weak common $p$-belief and $(\delta,\epsilon)$-common knowledge result in worse constants. For instance, converting weak common $p$-belief back and forth leads to weak common $(2p-1)$-belief.
\end{remark}

This implies that the two concepts can be applied in the same settings. An appealing feature of $(\delta,\epsilon)$-common knowledge is the equivalence of the $\sigma$-algebra-based definition to the two hierarchical definitions (Proposition  \ref{prop:equivalences_of_near_common_knowledge}), generalizing a characteristic of common knowledge (Proposition \ref{lemma:ck_in_two_ways}), which does not seem to be available for (weak) common $p$-belief (see, Morris \cite{Mo99}). A further advantage of $(\delta,\epsilon)$-common knowledge is that it is not restricted to the framework of countable probability spaces, but is defined for arbitrary probability spaces. This can be useful for extending results about the role of common $p$-belief for equilibrium selection in Bayesian games to Bayesian games with infinite type spaces.

\boldmath
\section{Approximate common knowledge of a random variable and agreement}
\unboldmath

Nielsen \cite{Ni84} extends the concept of common knowledge to random variables. Following this approach, we introduce a notion of $(\delta,\epsilon)$-common knowledge for random variables. We begin by reviewing Nielsen's account.

 \subsection{Nielsen's extension of common knowledge to random variables}

Nielsen introduces the following notions of knowledge and common knowledge of a random variable.

\begin{definition}{(Nielsen)}\label{def:understanding2}
	Let $X$ be a random variable and $\G_1, \G_2$ $\sigma$-algebras. The event  $K(\G_i,X)$ that \emph{$\G_i$ knows $X$}, for $i \in \{1,2\}$, is defined as the largest set $A \in \G_i$ such that $X\ind{A}$ is $\G_i$-measurable.
	The event $K(X)$ that \emph{$X$ is common knowledge} is defined as $K(\G_1 \cap \G_2,X)$.
\end{definition}

With these notions, Nielsen shows the following generalization of Aumann's Agreement Theorem (\ref{theo:Aumann}).

\begin{theorem}{(Nielsen)}\label{theo:aumann_nielsen}
	Let $X$ be a random variable and $\G_1,\G_2$ $\sigma$-algebras. For $i\in \{1,2\}$ we have $\E[X|\G_i]=\E[X| \G_1 \cap \G_2]$ 
	almost surely on $K(\E[X|\G_i])$. In particular, $\E[X|\G_1]=\E[X|\G_2]$ almost surely on $K(\E[X|\G_1]) \cap K(\E[X|\G_2])$.
\end{theorem}

As knowing a random variable $X$ is defined in terms of the measurability of $X$, we need to relax measurability to reach a relaxed notion of knowing $X$. To quantify how well $X$ is known by an individual with information $\G$---that is, how close it is to the $\G$-measurable posterior $\E[X|\G]$---we use the conditional variance $\Var(X|\G)= \E[(X-\E[X|\G])^2|\G]$. We begin by applying this approach to the special case of common knowledge of a random variable on the entire state space and establish an agreement result for this notion. Building on this, we introduce a relaxation of common knowledge of a random variable, $(\delta,\epsilon)$-common knowledge of a random variable, on an event $B$, and show that an approximate agreement result holds, which constitutes the main theorem of the section.

\boldmath
\subsection{Common information and $\epsilon$-common information}
\unboldmath

Nielsen introduces notions for the particular case that (common) knowledge of a random variable occurs on the entire state space.

\begin{definition}{(Nielsen)}\label{def:informed}
	Let $X$ be a random variable, $\G_1, \G_2$ $\sigma$-algebras and $i \in \{1,2\}$. If  $K(\G_i,X)=\Omega$,  
    then \emph{$\G_i$ is said to be informed about $X$}.
	If $K(X)=\Omega$, 
    then \emph{$X$ is said to be common information}.
\end{definition}

\begin{remark}\label{remark:common_information}
The property that a random variable $X$ is common information can equivalently be expressed through notions that only refer to a single individual, namely, that each individual is informed about $X$ (see also Remark \ref{remark: self-evident}).
\end{remark}

We begin by defining what it means for an individual to be $\epsilon$-informed about a random variable $X$ and what it means for $X$ to be $\epsilon$-common information.

\begin{definition}\label{def:common_information}
	Let $X \in  L_2(\Omega)$, $\G$ a $\sigma$-algebra, and $\epsilon \geq 0$. We say that \emph{ $\G$ is $\epsilon$-informed about $X$} if
	$$ \E[\Var(X|\G)] \leq \epsilon .$$	
\end{definition}

\begin{definition}\label{def-X-CK}
	
	Let $X\in L_2(\Omega)$ and $\G_1, \G_2$ $\sigma$-algebras. We say that \emph{$X$ is $\epsilon$-common information} if 
	$$\E[\Var(X|\G_i)] \leq \epsilon, $$
	for $i \in \{1,2\}$.
\end{definition}

Note that $0$-common information coincides with common information, as defined by Nielsen (Definition~\ref{def:informed}).

First, we show a preparatory lemma about the difference of posteriors.

\begin{lemma}\label{lemma:hilbert}
	Let $X \in L_2(\Omega)$ be a random variable and $\G_1,\G_2$ $\sigma$-algebras. 
	Then,
	\begin{align*}	
	 \E[(\E[X|\G_1]-\E[X|\G_2])^2] 
	\leq ||\E[X|\G_1]+ \E[X|\G_2]-\E[\E[X|\G_1]|\G_2]-\E[\E[X|\G_2]|\G_1]||_2 \cdot ||X||_2.\end{align*}
\end{lemma}

\begin{proof}
	Using several times that conditional expectations are self-adjoint, that is, 
	$$\E[Y \cdot \E[Z|\G]] = \E[\E[Y|\G] \cdot Z],$$
	for any random variables $Y,Z \in L_2(\Omega)$ and $\sigma$-algebra $\G$,
	 we calculate
	\begin{align*}
		&\E[(\E[X|\G_1]-\E[X|\G_2])^2] = \E[\E[X|\G_1]^2] + \E[\E[X|\G_2]^2] - 2 \E[\E[X|\G_1]\E[X|\G_2]]\\
		&=\E[ \E[X|\G_1] \cdot X] + \E[\E[X|\G_2] \cdot X] - \E[ \E[\E[X|\G_1]|\G_2] \cdot X] - \E[ \E[\E[X|\G_2]|\G_1] \cdot X] \\
		&=\E[(\E[X|\G_1] +  \E[X|\G_2] -  \E[\E[X|\G_1]|\G_2] -  \E[\E[X|\G_2]|\G_1]) \cdot X] \\
		&\leq ||\E[X|\G_1]+ \E[X|\G_2]-\E[\E[X|\G_1]|\G_2]-\E[\E[X|\G_2]|\G_1]||_2 \cdot ||X||_2,
	\end{align*}
where the last inequality is due to the Cauchy–Schwarz inequality.
\end{proof}

We next establish an Aumann-type result for the particular case of $\epsilon$-common information, showing that if the posteriors of a random variable $X$ are $\epsilon$-common information, then the expected squared difference of posteriors is at most $2\sqrt{\epsilon\Var(X)}$.

\begin{theorem}\label{theo:globalaumann}
	Let $\G_1, \G_2$ be $\sigma$-algebras, $X\in L_2(\Omega)$, $\epsilon \geq 0$. If the posteriors $  \E[X|\G_1]$ and $\E[X|\G_2]$ are both $\epsilon$-common information, then  
	$$\E [ (\E[X|\G_1]-\E[X|\G_2])^2  ] \leq 2\sqrt{\epsilon \Var(X)}.$$
\end{theorem}

\begin{proof}
	Without loss of generality, we may assume that $\E[X]=0$. Using Lemma~\ref{lemma:hilbert} and the triangle inequality in $L_2(\Omega)$, we find
	\begin{align*}
		|| \E[X|\G_1]- \E[X|\G_2]||_2^2 & \leq ||\E[X|\G_1] -  \E[ \E[X|\G_1] |\G_2] + \E[X|\G_2] -  \E[ \E[X|\G_2] |\G_1]||_2 \cdot ||X||_2  \\
		&\leq (\sqrt{\E [\Var( \E[X|\G_1]|\G_2)]} + \sqrt{\E [\Var( \E[X|\G_2]|\G_1)]}\ ) \cdot \sqrt{\Var (X)}  \\
		&\leq 2 \sqrt{\epsilon \Var(X)}.\qedhere
	\end{align*}
    
\end{proof}

Note that for random variables in $L_2(\Omega)$ that are common information, setting $\epsilon=0$ in Theorem~\ref{theo:globalaumann} recovers Nielsen's theorem (Theorem~\ref{theo:aumann_nielsen}).

\vspace{10pt}

\boldmath
\subsection{$(\delta,\epsilon)$-common knowledge of a random variable}\label{section: delta-epsilon-X}
\unboldmath

In this subsection, we give a notion of approximate common knowledge of a random variable \emph{on} an event B, show that this notion allows different equivalent formulations, and establish an Aumann-type result for it. To that end, we recall the definition of the restriction of a probability space.

\medskip
Let $(\Omega, \F, \P)$ be a probability space and $B$ an event, $\P(B)>0$. The probability space  $(B, \F_B, \P_B)$ is called the \emph{restriction of  $(\Omega, \F, \P)$ to $B$}, where $\F_B := \{ A \cap B : A \in \F \}$ is the trace $\sigma$-algebra of $\F$ in $B$ and ${\P_B(A):=\P(A\cap B)/\P(B)}$, for any event $A$. We further denote the restriction of a random variable $X$ to $B$ by $X_{|B}$. When referring to an individual's $\sigma$-algebra $\G$ on $(B, \F_B, \P_B)$, we implicitly refer to the trace $\sigma$-algebra $\G_B$.   
\medskip

With this, we can define $(\delta,\epsilon)$-knowledge and $(\delta,\epsilon)$-common knowledge of a random variable $X$ on an event $B$.   

\begin{definition}
	Let $X\in L_\infty(\Omega)$, $\G$ a $\sigma$-algebra, $B$ an event, $\P(B)>0$, and $\delta, \epsilon \geq 0$. We say that \emph{$X$ is $(\delta,\epsilon)$-known on $B$ by $\G$} if $B \in \G^\delta$ and $\G_{B}$ is $\epsilon$-informed about $X_{|B}$ on the space $(B, \F_B, \P_B)$, that is, $\E_{\P_B}[\Var_{\P_B}(X_{|B}|\G_{B})]\leq \epsilon$.
	For this property, we write $K_{ \delta, \epsilon}(\G, X, B).$
\end{definition}

\begin{definition}\label{def-X-deCK}
	Let $X\in L_\infty(\Omega)$, $\G_1, \G_2$  $\sigma$-algebras, $B$ an event, $\P(B)>0$, and $\delta, \epsilon \geq 0$. We say that \emph{$X$ is $(\delta,\epsilon)$-common knowledge on $B$} if $K_{\delta,\epsilon}(\G_1, X, B)$ and $K_{\delta,\epsilon}(\G_2, X, B).$
	For this notion, we write $K_{\delta,\epsilon}( X, B).$
	
\end{definition}

\begin{remark}
	On $B$, we have $$\E_{\P_B}[X_{|B}|\G_{B}]=\E[X|\sigma(\G \cup \{B\}) ].$$
\end{remark}

To establish alternative hierarchical characterizations and an agreement theorem for $(\delta,\epsilon)$-common knowledge of a random variable, we need to define a distance between complete $\sigma$-algebras. We adopt a notion introduced by Rogge~\cite{Ro74}. A slightly different, equivalent metric was previously studied by Boylan~\cite{Bo71} and by Neveu~\cite{Ne72}.

\begin{definition}
	Let $(\Omega, \F, \P)$ be a probability space. The set of all complete sub-$\sigma$-algebras of $\F$ is endowed with the metric $d$ given by
	$$
    d(\G_1,\G_2):= \max \left\{ \sup_{G_1 \in \G_1} \inf_{G_2\in\G_2} \P(G_1 \triangle G_2) , \sup_{G_2\in \G_2} \inf_{G_1\in\G_1} \P(G_1\triangle G_2)\right\},
    $$
	for complete $\sigma$-algebras $\G_1,\G_2 \subseteq \F$.
\end{definition}

An important property of this distance is the following result of Rogge \cite{Ro74}.

\begin{lemma}{(Rogge)}\label{lemma:rogge}
	Let $X$ be a random variable taking values in $[0,1]$ and $\G_1, \G_2$ sub-$\sigma$-algebras of $\F$. Then,
	$$  || \E[X|\G_1]-\E[X|\G_2]||_2 	\leq \sqrt{2 d(\G_1,\G_2)(1-d(\G_1,\G_2))}.$$
\end{lemma}

This shows that both conditional expectations and conditional variances of given bounded random variables are continuous with respect to the $\sigma$-algebra. In particular, 
	\begin{align}\label{align:bedingte_varianz_ist_stetig}
	||\Var(X|\G_1)-\Var(X|\G_2)||_2 \leq 4\sqrt{2d(\G_1,\G_2) (1-d(\G_1,\G_2))},
	\end{align}
for every random variable $X$ taking values in $[0,1]$ and $\sigma$-algebras $\G_1,\G_2$.

We proceed by collecting results on $\sigma$-algebras augmented by an additional event, which we will rely on in the proof of the main theorem of this section.

\begin{definition}
	Let $\A$ be a collection of events. We denote the smallest $\sigma$-algebra that contains $\A$ by $\sigma(\A)$.
\end{definition}

\begin{lemma}\label{lemma:adding_set_to_sigma_algebra}
	Let $\G$ be a $\sigma$-algebra and $A$ an event. Then,
	
	$$\sigma(\G \cup \{A\})= \{(G \cap A) \cup (H \cap A^\complement) | G,H\in \G\}. $$
\end{lemma}
\begin{proof}
    The result is obtained by a direct computation.
\end{proof}

\begin{lemma}\label{lemma:distance_for_sigma_algebras}
	Let $\G$ be a $\sigma$-algebra, $A$ an event, and $\epsilon \geq 0$ such that
	$$ \inf_{G\in \G} \P(A \triangle G) \leq \epsilon. $$
	Then,
	$$d(\G, \sigma(\G \cup \{A\})) \leq \epsilon.$$
\end{lemma}

\begin{proof}
	We begin by showing that 
	$$ \sup_{\tilde{G} \in \sigma(\G \cup \{A\})} \inf_{G'\in \G} \P(\tilde{G} \triangle G') \leq \epsilon.$$
	By Lemma~\ref{lemma:adding_set_to_sigma_algebra}, any event $\tilde{G} \in  \sigma(\G \cup \{A\})$ can be expressed as
	$$\tilde{G}=(G \cap A) \cup (H \cap A^\complement),$$
	for some $G,H \in \G.$ Furthermore, we know by Lemma~\ref{lemma:smallest_delta} that the event $B:=\{\P(A|\G)\geq \frac{1}{2}\}$ fulfills
		$$\P(A\triangle B) \leq \epsilon .$$
	We proceed by approximating $\tilde{G}$ by $(G \cap B) \cup (H \cap B^\complement) \in \G$ and find
	\begin{align*}
		&\inf_{G'\in \G} \P(((G \cap A) \cup (H \cap A^\complement) ) \triangle G')  \leq \P( ((G \cap A) \cup (H \cap A^\complement)) \triangle( (G \cap B) \cup (H \cap B^\complement) ) ) \\
		&\leq \P((( G \cap A) \triangle (G \cap B)) \cup ((H \cap A^\complement) \triangle (H \cap B^\complement)))  \\ 
		&\leq \P( (A\triangle B)  \cup (A^\complement \triangle B^\complement)) = \P(A\triangle B) \leq \epsilon.
	\end{align*}
	As $\tilde{G}$ was chosen arbitrarily, we conclude
	$$ \sup_{\tilde{G} \in \sigma(\G \cup \{A\})} \inf_{G'\in \G} \P(\tilde{G} \triangle G') \leq \epsilon.$$
	Finally, to conclude the proof, we note that
	$$\sup_{G'\in \G}\inf_{\tilde{G} \in \sigma(\G \cup \{A\})} \P(G' \triangle \tilde{G}) =0,$$
	as $\G \subseteq \sigma(\G \cup \{A\})$, and thus
	\[ d(\G, \sigma(\G \cup \{A\})) \leq \epsilon.\qedhere\]
\end{proof}

\begin{lemma}\label{lemma:distance_for_sigma_algebras_augmented_by_different_sets}
	Let $\G$ be a $\sigma$-algebra, $A,B$ events. Then, 
	$$d(\sigma(\G \cup \{A\}), \sigma(\G \cup \{B\})) \leq \P(A \triangle B ).$$
\end{lemma}

\begin{proof}
    We show that any event $A' \in \sigma(\G \cup \{A\})$ can be approximated by an event $B' \in \sigma(\G \cup \{B\})$ so that $\P(A' \triangle B') \leq \P(A \triangle B )$.
    Let $A' \in \sigma(\G \cup \{A\})$. By Lemma~\ref{lemma:adding_set_to_sigma_algebra} we know that there exist events $G, H \in \G$ so that $A'= (G \cap A) \cup (H \cap A^\complement)$. By the same lemma, we also know that the event $B'= (G \cap B) \cup (H \cap B^\complement)$ is an element of $\sigma(\G \cup \{B\})$. We estimate the symmetric difference between $A'$ and $B'$ by
	\begin{align*}
		\P(A' \triangle B')& = \P(( (G \cap A) \cup (H \cap A^\complement)) \triangle ((G \cap B) \cup (H \cap B^\complement))) \\
        &\leq \P(( (G \cap A) \triangle (G \cap B) ) \cup ( (H \cap A^\complement) \triangle (H \cap B^\complement) ) )\\
        & \leq \P((A \triangle B) \cup (A^\complement \triangle B^\complement)) = \P(A\triangle B).
	\end{align*}
    
    As this is true for every $A' \in  \sigma(\G \cup \{A\})$, we find $$ \sup_{A' \in \sigma(\G \cup \{A\})} \inf_{B'\in \sigma(\G \cup \{A\})} \P(A' \triangle B') \leq \P(A \triangle B ).$$

    The same inequality holds true with the roles of $A$ and $B$ reversed. Hence,  $d(\sigma(\G \cup \{A\}), \sigma(\G \cup \{B\})) \leq \P(A \triangle B ).$
\end{proof}

We proceed by establishing that our new notion of approximate common knowledge of a random variable $X$ admits alternative hierarchical characterizations.

\begin{definition}
	Let $\G_1,\G_2$ be $\sigma$-algebras, $X \in L_\infty(\Omega)$ a random variable, $B$ an event, $\P(B)>0$, and $\delta,\epsilon \geq 0$. Then, \emph{$X$ is hierarchically $(\delta, \epsilon)$-common knowledge on $B$} if there exist sequences $(C_n)_{n\in \N}$ and $(D_n)_{n\in \N}$ such that the following conditions hold for all $n \in \N$:
    \begin{itemize}
        \item $K_{\delta,\epsilon}(\G_1, X, C_n)$ and $K_{\delta,\epsilon}(\G_2,X,D_n)$,  
        \item $C_{n+1} \subseteq C_n \cap D_n$ and $D_{n+1} \subseteq C_n \cap D_n$, 
        \item $B=\bigcap_{n\in \N} \left( C_n \cap D_n \right)$.
    \end{itemize}
\end{definition}

\begin{definition}
	Let $\G_1,\G_2$ be $\sigma$-algebras, $X \in L_\infty(\Omega)$ a random variable, $B$ an event, $\P(B)>0$, and $\delta,\epsilon \geq 0$. Then,  \emph{$X$ is alternatingly hierarchically $(\delta,\epsilon)$-common knowledge on $B$} if there exists a sequence of events $(B_n)_{n\in \N}$ such that $B=\bigcap_{n\in \N} B_n$,  $B_i \subseteq B_j$ for $j < i$, and $K_{\delta,\epsilon}(\G_1, X,B_n)$ if $n$ is odd and  $K_{\delta,\epsilon}(\G_2, X,B_n)$ if $n$ is even.
\end{definition}

\begin{prop}\label{prop:equivalences_of_near_common_knowledge_X}
	Let $\G_1,\G_2$ be $\sigma$-algebras, $X \in L_\infty(\Omega)$ a random variable, $B$ an event, and $\delta, \epsilon \geq 0$. Then, the following conditions are equivalent:
	\begin{itemize}
		\item [(1)]   $X$ is $(\delta, \epsilon)$-common knowledge on $B$. 
		\item [(2)]  $X$ is hierarchically $(\delta, \epsilon)$-common knowledge on $B$. 
		\item [(3)]  $X$ is alternatingly hierarchically $(\delta, \epsilon)$-common knowledge on $B$.
	\end{itemize}
\end{prop}

\begin{proof}
    The implications $(1) \Rightarrow (2)$ and  $(2)  \Rightarrow (3)$ can be shown with the  same  proof as in Proposition~\ref{prop:equivalences_of_near_common_knowledge}.
    
    What is left to show is that $X$ being alternatingly hierarchically $(\delta, \epsilon)$-common knowledge on $B$ implies  $X$ being $(\delta, \epsilon)$-common knowledge on $B$. Let $(B_n)_{n\in \N}$ be a sequence of events as in the definition of alternating hierarchical $(\delta,\epsilon)$-common knowledge. Applying Lemma~\ref{lemma:falling_sequence} to the sequence $(B_{2n+1})_{n\in \N}$ and $A:=\Omega$, we see that $B \in \G_1^\delta$. The same argument applied to the sequence $(B_{2n})_{n\in \N}$ gives $B \in \G_2^\delta$ and thus $B \in \G_1^\delta \cap \G_2^\delta.$    

It remains to be shown that $\E_{\P_B}[\Var_{\P_B}(X_{|B}|\G_{i_{B}})] \leq \epsilon$ for $i \in \{1,2\}$. As $\lim_{n\to \infty} \P(B_n \triangle B) =0$, Lemma~\ref{lemma:rogge} together with Lemma~\ref{lemma:distance_for_sigma_algebras_augmented_by_different_sets} implies that
$$
\lim_{n\to \infty} \E[X | \sigma(\G_i \cup \{B_n\})]= \E[X|\sigma(\G_i \cup \{B\})]
$$
in $L_2(\Omega)$. From this, it is easy to see that also
\begin{align*}
&\lim_{n \to \infty} \E_{\P_{B_{n}}}[\Var_{\P_{B_{n}}}(X_{|B_{n}}|\G_{i_{B_n}})]=\lim_{n \to \infty} \frac{\E[(X-\E[X|\sigma(\G_i \cup \{B_n\})])^2\ind{B_n}]}{\P(B_n)}
\\
&=\frac{\E[(X-\E[X|\sigma(\G_i \cup \{B\})])^2\ind{B}]}{\P(B)}= \E_{\P_{B}}[\Var_{\P_{B}}(X_{|B}|\G_{i_{B}})]
\end{align*}
in $L_2(\Omega)$.
This concludes the proof as 
 $$ \E_{\P_B}[\Var_{\P_B}(X_{|B}|\G_{1_{B}})] = \lim_{n\to \infty} \E_{\P_{B_{2n+1}}}[\Var_{\P_{B_{2n+1}}}(X_{|B_{2n+1}}|\G_{1_{B_{2n+1}}})] \leq \epsilon  $$
    as well as
     \[ \E_{\P_B}[\Var_{\P_B}(X_{|B}|\G_{2_{B}})] = \lim_{n\to \infty} \E_{\P_{B_{2n}}}[\Var_{\P_{B_{2n}}}(X_{|B_{2n}}|\G_{2_{B_{2n}}})] \leq \epsilon. \qedhere \]
\end{proof}

We are ready to prove the main theorem of this section.

\begin{theorem}\label{theo:localaumann}
	Let $\G_1,\G_2$ be $\sigma$-algebras, $X\in  L_\infty(\Omega)$,  $B$ an event, $\P(B)>0$, and $\delta,\epsilon \geq 0$ such that $K_{\delta,\epsilon}(\E_{\P_B}[X|\G_1], B)$ and $K_{\delta,\epsilon}(\E_{\P_B}[X|\G_2], B).$ Then, the following inequality holds
	$$\E_{\P_B}[(\E[X|\G_1]-\E[X|\G_2])^2] \leq  64 \delta ||X||_{L_\infty(\P)}^2 +4 \sqrt{\epsilon} ||X||_{L_\infty(\P)}.$$
\end{theorem}

\begin{proof}
	Using the triangle inequality in $L_2(\P_B)$, we find
	
	\begin{align*}
		& ||\E[X|\G_1]-\E[X|\G_2]||_{L_2(\P_B)} \leq ||\E[X|\G_1]-\E_{\P_B}[X|\G_1]||_{L_2(\P_B)} \\
		&+  ||\E_{P_B}[X|\G_1]-\E_{P_B}[X|\G_2]||_{L_2(\P_B)} +||\E[X|\G_2]-\E_{\P_B}[X|\G_2]||_{L_2(\P_B)}.
	\end{align*}
	We first bound the first and the third term. For $i \in \{1,2\}$,
	\begin{align*}
		&||\E[X|\G_i]-\E_{\P_B}[X|\G_i]||_{L_2(\P_B)} = \frac{1}{\sqrt{\P(B)}} ||(\E[X|\G_i]-\E_{\P_B}[X|\G_i]) \ind{B}||_{L_2(\P)} \\ 
		&=\frac{1}{\sqrt{\P(B)}} ||(\E[X|\G_i]-\E[X|\sigma(\G_i \cup \{B\})] ) \ind{B}||_{L_2(\P)} \\
		&\leq \frac{1}{\sqrt{\P(B)}} ||\E[X|\G_i]-\E[X|\sigma(\G_i \cup \{B\})]  ||_{L_2(\P)} \\
		&\leq \frac{2||X||_{L_\infty(\P)}}{\sqrt{\P(B)}} \sqrt{2\delta \P(B)}= 2\sqrt{2\delta}||X||_{L_\infty(\P)} ,
	\end{align*}
	where the last inequality is due to Lemma~\ref{lemma:rogge} and Lemma~\ref{lemma:distance_for_sigma_algebras}.
	
	For the second term, we use that, by assumption,  $K_{\delta,\epsilon}(\E_{\P_B}[X|\G_1], B)$ and $K_{\delta,\epsilon}(\E_{\P_B}[X|\G_2], B)$, and so we can apply Theorem~\ref{theo:globalaumann}, and get
	$$||\E_{P_B}[X|\G_1]-\E_{P_B}[X|\G_2]||_{L_2(\P_B)} \leq \sqrt{2\sqrt{\epsilon} ||X||_{L_2(\P_B)}} \leq \sqrt{2\sqrt{\epsilon} ||X||_{L_\infty(\P)}}. $$
	
	What is left to do is to bring the estimates above together. For this, we use the fact that $(a+b)^2 \leq 2(a^2+b^2)$ for all real numbers $a,b$, and find
	\begin{align*}
	&\E_{\P_B}[(\E[X|\G_1]-\E[X|\G_2])^2] \leq \left(4\sqrt{2\delta} ||X||_{L_\infty(\P)} +\sqrt{2\sqrt{\epsilon} ||X||_{L_\infty(\P)}}\right)^2 \\
	&\leq 64\delta ||X||_{L_\infty(\P)}^2 + 4 \sqrt{\epsilon}||X||_{L_\infty(\P)}.\qedhere
	\end{align*} 
\end{proof}

Theorem~\ref{theo:localaumann} is a local version of Theorem~\ref{theo:globalaumann} in the following sense: Let both individuals pick an event $B \in \G_1^\delta \cap \G_2^\delta$ of which they assume that it occurs, and let them calculate their posteriors of a random variable $X$ given $B$. If these posteriors, while assuming $B$, are $\epsilon$-common information, then the $L_2$-distance on $B$ of the true posteriors (without assuming $B$ as given) is bounded as a function of $\delta$ and $\epsilon$.

\boldmath
\section{Dialogues and ``almost'' convergence of posteriors}
\unboldmath

An Aumann-type result for dynamic settings, in which individuals continue to exchange information of their posteriors of a given random variable $X$, is given by Geanakoplos and Polemarchakis \cite{GePo82} for discrete probability spaces. Nielsen \cite{Ni84} generalizes this result to arbitrary probability spaces.

\subsection{Bayesian Dialogues over a random variable}

 To model the passing of time and the individuals' ability to learn new information, following Nielsen, we associate each individual $i$ with a complete filtration $(\G^i_n)_{n\in \N}$. Furthermore, we set $\G^i_\infty:= \sigma(\bigcup_{i=1}^\infty \G^i_n)$.

\begin{theorem}(Nielsen)\label{theo:dialog_nielsen}
		Let $X\in L_1(\Omega)$ and filtrations $(\G_n^1)_{n \in \N}$ and $(\G_n^2)_{n\in \N}$ be given, such that for each $n\in \N$, there exists an $m>n$ such that $\E[X|\G^1_n]$ and  $\E[X|\G^2_n]$  are common information at time $m$. Then, $$\lim_n \E[X|\G_n^1] = \lim_n \E[X|\G_n^2]  \, a.s.$$
\end{theorem}

In other words, if at each point in time $n$, for each individual $i$, the other individual will, at some later point in time,  $m>n$, know all information pertaining to $X$ that $i$ knew at time $n$, then their posteriors converge to the same value, almost surely.

The condition that each individual eventually learns the posterior of the other can be expressed in terms of conditional variances, namely, that for $i,j \in \{1,2\}$ and $n \in \N$, there exists $N_n \in \N$ such that for all $m >N_n$, we have $\E[\Var(\E[X|\G^i_n]|\G^j_m)]=0.$

The following theorem uses this condition to establish a relaxed version of Nielsen's result: If each individual approximately knows the posterior of the other individual in the limit as time progresses, then their posteriors at infinity are approximately the same.

\begin{theorem}\label{theo:dialoge_epsilon}
	Let $X \in L_2(\Omega)$ and filtrations $(\G^1_n)_{n \in \N}$ and $(\G^2_n)_{n \in \N}$ be given. Assume that for all $i,j \in \{1,2\}$
	\begin{align}\label{align:dialoge_epsilon}
	\lim_{n\to\infty} \lim_{m\to\infty} \E[ \Var(\E[X |\G^i_n]|\G^j_m)] \le \epsilon.
	\end{align}
	Then,
	$$  \E[( \E[X|\G^1_\infty] -  \E[X|  \G^2_\infty] )^2] \le  2\sqrt{\epsilon \Var(X)}.$$
\end{theorem}

\begin{proof}
Let $i,j \in \{1,2\}$. 
By Doob's martingale convergence theorem,
$$\lim_{m\to\infty} \E[ \E[X|\G^i_n]|\G^j_m] =  \E[ \E[X|\G^i_n]|\G^j_\infty] $$
 in $L_2(\Omega)$, for all $n \in \N$, and hence,
 \begin{align}\label{equ:dialoge_eps1}
 	\begin{split}
 	\lim_{m\to\infty}  \E[ \Var(\E[X |\G^i_n]|\G^j_m)]  &= \lim_{m\to\infty} || \E[X|\G^i_n] - \E[ \E[X|\G^i_n]|\G^j_m] ||^2_2 \\ &= || \E[X|\G^i_n] - \E[ \E[X|\G^i_n]|\G^j_\infty] ||^2_2.
 \end{split}
 \end{align}

Applying Doob's martingale inequality once more to $\E[X|\G^i_n]$, we see
\begin{align}\label{equ:dialoge_eps2}
	\lim_{n\to\infty} \E[X|\G^i_n] = \E[X|\G^i_\infty]
\end{align}
in $L_2(\Omega)$. This, together with Jensen's inequality (see, for instance, Williams \cite[p.\ 88]{Wi91}), implies that
\begin{align*}
	\begin{split}
	&\lim_{n\to\infty} || \E[ \E[X|\G^i_n]|\G^j_\infty] -  \E[ \E[X|\G^i_\infty]|\G^j_\infty] ||^2_2 = \lim_{n\to\infty} \E[ (\E[\E[X|\G^i_n] - \E[X|\G^i_\infty] |\G^j_\infty])^2] \\ \leq & \lim_{n\to\infty} \E[ \E[(\E[X|\G^i_n] - \E[X|\G^i_\infty])^2 |\G^j_\infty]]  = \lim_{n\to\infty} ||\E[X|\G^i_n] - \E[X|\G^i_\infty] ||^2_2 =0,
	\end{split}
\end{align*}
and hence
\begin{align}\label{equ:dialoge_eps3}
\lim_{n\to\infty}  \E[ \E[X|\G^i_n]|\G^j_\infty] =  \E[ \E[X|\G^i_\infty]|\G^j_\infty]
\end{align}
in $L_2(\Omega)$.
Next, we use inequality~\eqref{align:dialoge_epsilon}, together with equations~\eqref{equ:dialoge_eps1}, \eqref{equ:dialoge_eps2} and~\eqref{equ:dialoge_eps3} to calculate
\begin{align*}
	\epsilon \geq &\lim_{n\to\infty} \lim_{m\to\infty}  \E[ \Var(\E[X |\G^i_n]|\G^j_m)] = \lim_{n\to\infty} || \E[X|\G^i_n] - \E[ \E[X|\G^i_n]|\G^j_\infty] ||^2_2 \\ =& || \E[X|\G^i_\infty] - \E[ \E[X|\G^i_\infty]|\G^j_\infty] ||^2_2 = \E[\Var(\E[X|\G^i_\infty]|\G^j_\infty)] ,
\end{align*}
which means that $\E[X|\G^i_\infty]$ is $\epsilon$-common information for all $i \in \{1,2\}$. Applying Theorem~\ref{theo:globalaumann} to $\E[X|\G^1_\infty]$ and $\E[X|\G^2_\infty]$ proves that
\[ \E[( \E[X|\G^1_\infty] -  \E[X|\G^2_\infty] )^2] \le  2\sqrt{\epsilon \Var{(X)}}. \qedhere \]
\end{proof}

If $\epsilon=0$, we also get that the posteriors converge a.s.\ to each other.

\begin{corollary}
			Let $X \in L_2(\Omega)$ and $(\G^1_n)_{n\in \N}, (\G^2_n)_{n\in \N}$ be filtrations. Furthermore, assume that 
	$$
	\lim_n \lim_m \E[\Var(\E[X |\G^1_n]|\G^2_m)] =0	\ \text{and } \	\lim_n \lim_m \E[\Var(\E[X |\G^2_n]|\G^1_m)] =0.
	$$
	Then, $$\lim_n\E[X|\G^1_n]= \lim_n\E[X|\G^2_n] \, a.s.$$

\end{corollary}

\begin{proof}
By Theorem~\ref{theo:dialoge_epsilon}, 
$$\E[( \E[X|\G^1_\infty]-  \E[X|\G^2_\infty]  )^2]=0,$$
hence 
 $  \E[X|\G^1_\infty] = \E[X|\G^2_\infty]$ a.s.
 By Doob's martingale convergence theorem, we know that
 $$\lim_{n\to\infty} \E[X|\G^1_n] = \E[X|\G^1_\infty] \, a.s.$$
Similarly for $(\G^2_n)_{n\in \N}$, and thus
 \[ \lim_{n\to\infty} \E[X|\G^1_n] = \E[X|\G^1_\infty] = \E[X|\G^2_\infty] = \lim_{n\to\infty} \E[X|\G^2_n] \, a.s. \qedhere \] 
\end{proof}

\subsection{Communication with noise}

Theorem~\ref{theo:dialoge_epsilon} has the following important application.

\begin{example}\label{ex:com_noise}
Consider communication between two individuals, individual 1 and individual 2, through a noisy channel, where the individuals tell each other their current posteriors of a given random variable $X$. The noise in the channel is given by the random variables $(\eta^1_n)_{n\in\N}$ and $(\eta^2_n)_{n\in\N}$ with a uniform bound on their variances, $\Var(\eta^i_n) \leq \epsilon$ for $i \in \{1,2\}$, $n \in \N$ and some $\epsilon > 0$.

The individuals start with the $\sigma$-algebras $\G^1_0,\G^2_0$, encoding their knowledge before communication starts, and update their knowledge as follows:
$$
\G^1_{n+1} = \G_n^1 \vee \sigma ( \E[X|\G^2_n] + \eta_n^2 ),
$$	
and analogously for individual 2. 
For natural numbers $n \leq m $, we calculate
	\begin{align*}
\E[ \Var( \E[X | \G^i_n] | \G^j_{m} )] &\leq \E[ \Var( \E[X | \G^i_n] | \G^j_{n+1} )] \le \E [ \Var( \E[X | \G^i_n] | \E[X|\G^i_n] 
 + \eta_n^i )] \\ &=\E[ \Var(\eta_n^i | \E[X|\G^i_n] + \eta_n^i )] \le \Var(\eta_n^i) \le \epsilon.
	\end{align*}
Applying Theorem~\ref{theo:dialoge_epsilon}, we see that the difference in the individuals' posteriors at infinity is controlled by $\epsilon$.
\end{example}

\begin{remark}
	 In this example, we did not need to assume anything about the shape of the noise, only that its variance is uniformly bounded.
\end{remark}

\vspace{1em}

\noindent
{\bf Acknowledgements.}
We are grateful to Mathias Beiglböck for his comments throughout this work and for initiating the collaboration that resulted in this paper. We also benefited from discussions with Michael Greinecker. C.P. gratefully acknowledges the hospitality of the Faculty of Mathematics at the University of Vienna.

This research was funded in whole or in part by the Austrian Science
Fund (FWF)  10.55776/ P34743 and 10.55776/J4981 as well as the Austrian National Bank
[Jubiläumsfond, project 18983]. For open access purposes, the author has
applied a CC BY public copyright license to any author accepted
manuscript version arising from this submission.

\bibliographystyle{abbrv} 

\bibliography{joint_biblio}
\end{document}